\documentclass[
  aps,
  pra,
  reprint,
  superscriptaddress
]{revtex4-2}

\usepackage[british]{babel}

\usepackage[T1]{fontenc}
\usepackage{libertinus}
\usepackage{libertinust1math}

\usepackage{mathtools}
\usepackage{amsthm}
\usepackage{amssymb}
\usepackage{bm}
\usepackage{siunitx}

\usepackage{graphicx}
\usepackage{tikz}
\usepackage[dvipsnames]{xcolor}

\usepackage[
  breaklinks=true,
  colorlinks=true,
  linkcolor=myblue,
  urlcolor=myblue,
  citecolor=myblue
]{hyperref}
\usepackage{orcidlink}

\makeatletter
\newcommand*{\transpose}{%
  {\mathpalette\@transpose{}}%
}
\newcommand*{\@transpose}[2]{%
  \raisebox{\depth}{$\m@th#1\intercal$}%
}
\makeatother

\newcommand{\bs}[1]{\boldsymbol{#1}}
\newcommand{\dd}{\mathop{}\!\mathrm{d}}

\newcommand{\Tr}{\operatorname{Tr}}

\newcommand{\PGM}{\mathrm{PGM}}
\newcommand{\opt}{\mathrm{opt}}
\newcommand{\IF}{\mathrm{IF}}

\newtheorem{theorem}{Theorem}

\newtheorem{lemma}{Lemma}

\newtheorem{remark}{Remark}
\definecolor{myblue}{RGB}{34,31,150}
\definecolor{mygreen}{RGB}{34,150,31}
\definecolor{myred}{RGB}{190,30,30}

\begin{document}

\title{Measurement incompatibility in Bayesian quantum estimation beyond quadratic loss}

\author{Francesco Albarelli\,\orcidlink{0000-0001-5775-168X}}
\email{francesco.albarelli@gmail.com}
\affiliation{Università di Parma, Dipartimento di Scienze Matematiche, Fisiche e Informatiche, I-43124 Parma, Italy}
\affiliation{INFN—Sezione di Milano-Bicocca, Gruppo Collegato di Parma, I-43124 Parma, Italy}

\author{Stefano Chessa\,\orcidlink{0000-0003-2771-8330}}
\affiliation{Scuola Normale Superiore, 56126 Pisa, Italy}

\author{Jes\'us Rubio\,\orcidlink{0000-0002-8193-8273}}
\affiliation{Department of Physics and Astronomy, University of Exeter, Stocker Road, Exeter EX4 4QL, United Kingdom}

\date{20th October 2026}

\begin{abstract}
In single-shot Bayesian multiparameter quantum estimation, measurement incompatibility can increase the minimum joint quadratic loss by at most a factor of two relative to the bound obtained by optimising each parameter separately.
Using the generalised Barnum--Knill inequality, we establish the same bound for a wider class of additive losses generated by positive score functions, showing that this limitation extends beyond the standard quadratic setting.
To illustrate this result, we consider a broad family of periodic losses, including the sine loss sometimes employed in interferometry, and derive the corresponding incompatibility bound for Bayesian multiphase estimation with arbitrary prior knowledge.
We further extend the framework to unbounded losses, such as absolute losses; this also recovers the case of quadratic losses.
\end{abstract}

\maketitle

\section{Introduction}

Measurement incompatibility, the fact that measurements individually optimal for different parameters cannot always be performed simultaneously, is a fundamental limitation in multiparameter quantum estimation theory (QET)~\cite{Albarelli2020a,Demkowicz-Dobrzanski2020,Pezze2025,Conlon2026a}.
A related question also arises in hypothesis testing for the simultaneous discrimination of multiple quantum states~\cite{Audenaert2014a,liDiscriminatingQuantumStates2016,Cheng2026,Lee2026b}.
Until recently, its quantitative impact on estimation errors had been investigated predominantly within local QET, where the parameters of interest are assumed to be approximately known and precision is characterised through Cramér--Rao-type bounds~\cite{Belliardo2021,Sidhu2021,Candeloro2021c,Chen2022f,Hayashi2023a,Candeloro2024,Chen2024j,Das2025b,He2026d,He2026c,Wang2026n}.
A work by some of us~\cite{Albarelli2026c} has introduced a treatment of this question within the Bayesian paradigm, which incorporates prior information and applies beyond local and asymptotic regimes.

In this framework, measurement incompatibility can be quantified relative to the symmetric posterior mean (SPM) bound.
The SPM bound represents the mean quadratic loss, with individual losses of the form $(\tilde{\theta} - \theta)^2$ for each parameter $\theta$ and its estimate $\tilde{\theta}$, attainable if the individually optimal measurements were jointly implementable.
Introduced in Ref.~\cite{Rubio2020a} (see also Refs.~\cite{Sidhu2020,Demkowicz-Dobrzanski2020}), it plays a role analogous to that of the Helstrom Cramér--Rao bound (CRB) based on the symmetric logarithmic derivative (SLD) in local QET~\cite{Helstrom1976,Holevo2011b,Paris2009,Liu2020}.
Since then, tighter lower bounds, such as the Bayesian Nagaoka--Hayashi bound, have been developed~\cite{Suzuki2024a}, and the topic has attracted renewed attention~\cite{Zhang2024g,Zhang2026n,Zhao2026c,Bia2026,Law2026}.

Measurement incompatibility can be shown to increase the minimum loss above the SPM benchmark by at most a factor of two, a result obtained by considering an explicit estimation strategy through the so-called pretty good measurement (PGM)~\cite{Albarelli2026c}.
This factor of two parallels the result in local QET, where the Holevo CRB, which represents the fundamental asymptotically attainable weighted sum of estimator variances for all parameters (under suitable regularity conditions~\cite{Demkowicz-Dobrzanski2020,Yang2019}), is at most twice the Helstrom SLD-CRB~\cite{Carollo2019,Tsang2020a}.

Yet, the quadratic loss for which this result was established, both in Bayesian and local QET, is not always the appropriate figure of merit.
Quadratic losses can accommodate a wide range of loss functions beyond the standard square error through the location-isomorphic mapping introduced in symmetry-informed estimation~\cite{Rubio2024}.
However, such transformations do not capture, e.g., the global structure of periodic parameter spaces.
A prominent example is phase estimation, where the underlying $U(1)$ symmetry of each parameter naturally calls for trigonometric losses~\cite{Holevo2011b}, such as the sine loss~\cite{Demkowicz-Dobrzanski2011, Helstrom1976}, which reduces to the standard quadratic loss for small errors but is different otherwise.
This raises the question of whether the factor-of-two limitation on measurement incompatibility is specific to quadratic loss or holds more generally.

Here we show that it extends to \emph{any} additive loss whose individual contributions are generated by positive score functions.
Our proof builds on the generalised Barnum--Knill (BK) inequality of Mishra, Lami, and Wilde~\cite{Mishra2025}.
Specifically, we show that the same PGM simultaneously incurs, for every parameter, at most twice the minimum loss achievable by optimising the measurement for that parameter alone.
The additivity of the total loss then yields the same bound on measurement incompatibility.
This ability of the PGM to extract information about multiple quantities simultaneously has also been explored in other quantum information tasks~\cite{Lin2025a,Grewal2026a}.
We first apply this result to a broad class of periodic Holevo losses, including the sine loss, establishing the corresponding incompatibility bound for Bayesian multiphase estimation with arbitrary prior distributions.
We further extend the result to unbounded losses through positive exponential scores; we illustrate this extension by applying it to absolute losses and use it to recover the result for quadratic losses.

In short, our results show that, in Bayesian multiparameter QET, the positivity of the score functions associated with additive losses is a sufficient condition to limit the increase in minimum joint loss due to measurement incompatibility to at most a factor of two relative to optimising each parameter separately, even when the loss does not have a simple quadratic form.

\section{Bayesian multiparameter QET}
\label{sec:bayesian}

We start by summarising the Bayesian approach to multiparameter QET and define the associated incompatibility-free benchmark.
We follow the notational conventions introduced in Ref.~\cite{Albarelli2026c}.

Let $\bs{\theta}\in\Omega = \Omega_1 \times \cdots \times \Omega_d$ denote a vector of hypotheses for a corresponding set of $d$ unknown parameters, with prior density $p(\bs{\theta})$.
Further consider a parametric family of states $\rho(\bs{\theta})$ defined on a Hilbert space~\footnote{The generalised BK inequality on which our results are based holds also for infinite-dimensional Hilbert spaces~\cite{Mishra2025}.
Thus, we expect our results to extend to this setting under mild regularity assumptions, which however we do not work out in detail here, favouring a more concise and pedagogical presentation.}.
A measurement of $\rho(\boldsymbol{\theta})$ can be described by a positive operator-valued measure (POVM) with elements $\lbrace M(\bs{\tilde{\theta}}) \rbrace$ labelled directly by the estimates $\bs{\tilde{\theta}}\in\Omega$.

To quantify the decision-theoretic loss associated with reporting a vector estimate $\bs{\tilde{\theta}}$ when the true parameter values are $\bs{\theta}$, we introduce an \emph{additive} loss function
\begin{equation}
    \ell(\boldsymbol{\tilde{\theta}},\boldsymbol{\theta})
    =
    \sum_{j=1}^d w_j\ell_j(\tilde{\theta}_j,\theta_j),
    \quad w_j\geq0,
    \quad \ell_j\geq0.
    \label{eq:additive-loss}
\end{equation}
This recovers the quadratic loss with a diagonal weight matrix of elements $W_{jk}=w_j\delta_{jk}$, by taking $\ell_j(\tilde{\theta}_j,\theta_j)=(\tilde{\theta}_j-\theta_j)^2$, but is considerably more general.

The \emph{mean loss} associated with Eq.~\eqref{eq:additive-loss} reads 
\begin{equation}
    \mathcal L(M) \coloneqq \sum_{j=1}^d w_j\mathcal L_j(M),
\end{equation}
where
\begin{equation}
    \mathcal L_j(M)
    =\int\dd\boldsymbol{\theta}
    \dd\boldsymbol{\tilde{\theta}}\,
    p(\boldsymbol{\theta})
    \Tr[\rho(\boldsymbol{\theta})M(\boldsymbol{\tilde{\theta}})]
    \ell_j(\tilde{\theta}_j,\theta_j).
\end{equation}
The minimum loss is defined as $\mathcal L_{\opt}\coloneqq \min_M\mathcal L(M)$, with the minimising POVM corresponding to the jointly optimal strategy.

To isolate measurement incompatibility, consider first, for each parameter, the single-parameter optimum
\begin{equation}
    \mathcal L_j^*
    =\min_{M_j}
    \int\dd\boldsymbol{\theta}
    \dd\tilde{\theta}_j\,
    p(\boldsymbol{\theta})
    \Tr[\rho(\boldsymbol{\theta})M_j(\tilde{\theta}_j)]
    \ell_j(\tilde{\theta}_j,\theta_j),
    \label{eq:single-par_opt_loss}
\end{equation}
where $M_j$ ranges over POVMs with outcomes in $\Omega_j$.
In each of these estimation problems, the components of $\boldsymbol{\theta}$ other than $\theta_j$ are treated as unknown nuisance parameters~\cite{Suzuki2020b} and can therefore be marginalised over.
Since $\ell_j$ depends only on $\theta_j$, we see that $\mathcal L_j^*$ is the minimum single-parameter loss obtained from the marginal prior $p(\theta_j)$ and from the one-parameter family of averaged states
\begin{equation}
    \rho_j(\theta_j)
    =
    \int\prod_{k\neq j}\dd\theta_k\,
    p(\boldsymbol{\theta}_{\neg j}|\theta_j)\,
    \rho(\boldsymbol{\theta}),
    \label{eq:averaged_states}
\end{equation}
where $\boldsymbol{\theta}_{\neg j}$ collects all components except $\theta_j$.
Any correlations between $\theta_j$ and the other parameters, both through the prior and through the states $\rho(\boldsymbol{\theta})$, are thus fully taken into account in the quantity $\mathcal L_j^*$.

Any joint POVM $M$ defines, for each $j$, a POVM on $\Omega_j$ through its $j$th marginal
\begin{equation}
    M_j(\tilde{\theta}_j)
    =
    \int\prod_{k\neq j}\dd\tilde{\theta}_k\,
    M(\boldsymbol{\tilde{\theta}}),
    \label{eq:povm_marginal}
\end{equation}
which amounts to performing the measurement $M$, but retaining only the estimate of $\theta_j$.
Since $\ell_j$ depends on $\boldsymbol{\tilde{\theta}}$ only through $\tilde{\theta}_j$, the loss $\mathcal L_j(M)$ depends on $M$ only through $M_j$.
Conversely, every POVM $M_j$ on $\Omega_j$ is the $j$th marginal of some joint POVM, e.g. of
$M(\boldsymbol{\tilde{\theta}})=M_j(\tilde{\theta}_j)\prod_{k\neq j}q_k(\tilde{\theta}_k)$
for arbitrary probability densities $q_k$ on $\Omega_k$;
this amounts to reporting the remaining estimates
at random, independently of the measurement outcome.

We then define the incompatibility-free benchmark as
\begin{equation}
\label{eq:LIF}
    \mathcal L_{\IF} \coloneq \sum_{j=1}^d w_j\mathcal L_j^*,
\end{equation}
namely, the loss attainable if the measurements individually optimal for each parameter, with the others unknown, were jointly implementable.
This recovers Personick's minimum mean square loss~\cite{Personick1971} and leads to the SPM bound~\cite{Rubio2020a} for quadratic loss with a diagonal weight matrix.

\begin{lemma}\label{lem:if-benchmark}
The incompatibility-free benchmark satisfies
\begin{equation}
    \mathcal L_{\IF}\leq\mathcal L_{\opt}.
\end{equation}
\end{lemma}

\begin{proof}
By Eq.~\eqref{eq:povm_marginal}, $\mathcal L_j(M)$ equals the single-parameter loss of $M_j$, hence $\mathcal L_j^*\le\mathcal L_j(M)$.
Multiplying by $w_j$, summing over $j$, and taking the minimum over $M$ proves the claim.
\end{proof}

Finally, we recall the incompatibility quantifier introduced in Ref.~\cite{Albarelli2026c},
\begin{equation}
    \mathcal I
    =
    \frac{\mathcal L_{\opt}}{\mathcal L_{\IF}}-1.
    \label{eq:incompatibility_ratio}
\end{equation}
This may be interpreted as the fractional increase in the minimum loss due to measurement incompatibility.
In the next section, we present our main result in terms of both $\mathcal L$ and $\mathcal I$.

\section{Incompatibility bounds for additive losses}
\label{sec:additive-losses}

In this section, we establish an upper bound on the effect of measurement incompatibility for a broad class of the additive losses introduced in Sec.~\ref{sec:bayesian}.
We begin with bounded losses and formulate the relevant condition in terms of \emph{score functions}.

A score $\sigma:Z\times Z\to[0,1]$, for a generic parameter space $Z$, quantifies the agreement between a true value and an estimate, with larger values corresponding to better agreement.
Consider positive scores, namely, those admitting a self-convolution factorisation
\begin{equation}
    \label{eq:sigma_score_def}
    \sigma(z,z')
    =
    \int \dd t\,\mu(t) h(z,t)h(z',t)
    \equiv
    (h\star h)(z,z'),
\end{equation}
with measure $\mu$ and a real and bounded function $h$, and losses of the form
\begin{equation}
    \ell_j(\tilde{\theta}_j,\theta_j)
    =
    c_j[1-\sigma_j(\tilde{\theta}_j,\theta_j)],
    \quad c_j>0.
    \label{eq:score_loss}
\end{equation}
For this class, minimising each contribution to the mean loss is equivalent to maximising the corresponding expected score,
\begin{equation}
    \varsigma_j(M)
    \coloneq
    \int\dd\boldsymbol{\theta}\,
    \dd\boldsymbol{\tilde{\theta}}\,
    p(\boldsymbol{\theta})
    \Tr[
        \rho(\boldsymbol{\theta})
        M(\boldsymbol{\tilde{\theta}})
    ]
    \sigma_j(\tilde{\theta}_j,\theta_j).
    \label{eq:expected_score}
\end{equation}

We next invoke the generalised BK inequality~\cite{Mishra2025}, which relates the optimal expected score to that attained by the square-root measurement, also known as PGM~\cite{Belavkin1975,Hausladen1994}.
Let
\begin{equation}
    \rho
    =
    \int\dd\boldsymbol{\theta}\,
    p(\boldsymbol{\theta})
    \rho(\boldsymbol{\theta})
\end{equation}
be the prior-averaged state, with $\rho\succ0$ (otherwise, all expressions are understood to be restricted to $\operatorname{supp}\rho$).
The corresponding PGM is
\begin{equation}
    M_{\PGM}(\boldsymbol{\tilde{\theta}})
    =
    p(\boldsymbol{\tilde{\theta}})
    \rho^{-1/2}
    \rho(\boldsymbol{\tilde{\theta}})
    \rho^{-1/2}.
    \label{eq:pgm}
\end{equation}
For a positive score $\sigma$, let $\varsigma^*$ denote its optimal expected value and $\varsigma(M_{\PGM})$ that attained by the PGM.
The generalised BK inequality then states that
\begin{equation}
    \varsigma^*
    \leq
    \sqrt{\varsigma(M_{\PGM})}.
    \label{eq:generalised_BK}
\end{equation}
We can now apply this result independently to each contribution to the additive loss, i.e., to each marginal of the PGM.
A similar idea has been exploited before in a different context~\cite{Lin2025a,Grewal2026a}.

\begin{theorem}\label{thm:factor_two}
Suppose that, for every $j=1,\ldots,d$, there is a constant $c_j>0$ and a positive score $\sigma_j:\Omega_j\times\Omega_j\to[0,1]$ such that Eq.~\eqref{eq:score_loss} holds.
Then
\begin{equation}
    \mathcal L_{\IF}
    \leq
    \mathcal L_{\opt}
    \leq
    \mathcal L(M_{\PGM})
    \leq
    2\mathcal L_{\IF}.
    \label{eq:factor_two}
\end{equation}
Consequently, whenever $\mathcal L_{\IF}>0$, the incompatibility quantifier defined in Eq.~\eqref{eq:incompatibility_ratio} satisfies
\begin{equation}
    0\leq\mathcal I\leq1.
\end{equation}
\end{theorem}

\begin{proof}
The first two inequalities in Eq.~\eqref{eq:factor_two} follow, respectively, from Lem.~\ref{lem:if-benchmark} and from the fact that $M_{\PGM}$ is a feasible joint measurement.
It remains to prove the last inequality.

Fix $j$ and extend $\sigma_j$ to a score on the full parameter space according to
\begin{equation}
    \sigma'_j
    (\boldsymbol{\tilde{\theta}},\boldsymbol{\theta})
    =
    \sigma_j(\tilde{\theta}_j,\theta_j).
\end{equation}
This construction extends $\sigma_j$ from $\Omega_j$ to $\Omega$ by assigning the same score to any two parameter vectors that share the same $j$th components, irrespective of their remaining components.
If $\sigma_j=h_j\star h_j$, then 
\begin{equation}
    \sigma'_j=h'_j\star h'_j,
    \quad
    h'_j(\boldsymbol{\tilde{\theta}},t)
    =
    h_j(\tilde{\theta}_j,t).
\end{equation}
Thus $\sigma'_j$ is a positive score on $\Omega$.

Now, let $\varsigma_j^*$ denote the optimal expected score for estimating the single component $\theta_j$.
Equivalently, $\varsigma_j^*$ is the optimal expected value of the extended score $\sigma'_j$, since an $\Omega$-outcome POVM contributes only through its $j$th marginal, while any $\Omega_j$-outcome POVM can be extended to an $\Omega$-outcome POVM by adding arbitrary classical randomness in the remaining coordinates.

Applying Eq.~\eqref{eq:generalised_BK} to $\sigma'_j$ gives
\begin{equation}
    \varsigma_j^*
    \leq
    \sqrt{\varsigma_j(M_{\PGM})}.
\end{equation}
Using Eq.~\eqref{eq:score_loss},
\begin{align}
    \mathcal L_j(M_{\PGM})
    &=
    c_j[1-\varsigma_j(M_{\PGM})]
    \nonumber\\
    &\leq
    c_j[1-(\varsigma_j^*)^2]
    \nonumber\\
    &=
    c_j(1-\varsigma_j^*)
    (1+\varsigma_j^*)
    \nonumber\\
    &\leq
    2c_j(1-\varsigma_j^*)
    =
    2\mathcal L_j^*.
    \label{eq:ineq_chain_PGM_factortwo}
\end{align}
Multiplying by $w_j$ and summing over $j$ gives
\begin{equation}
    \mathcal L(M_{\PGM})
    \leq
    2\mathcal L_{\IF},
\end{equation}
which completes the proof.
\end{proof}

\begin{remark}[Optimised post-processing]
\label{remark:optimised-post-procesing}
The bounds above use the PGM with its canonical outcome labels $\boldsymbol{\tilde{\theta}}$.
If the same quantum POVM is followed by an optimal classical estimator for the chosen loss, the resulting mean loss can only decrease.
Hence, Eq.~\eqref{eq:factor_two} remains valid, a fortiori, when the PGM is followed by optimal classical post-processing.
\end{remark}

\section{Holevo losses in multiphase estimation}
\label{sec:multiphase}

We now apply Thm.~\ref{thm:factor_two} to the paradigmatic case of multiphase estimation~\cite{Macchiavello2003,Barbieri2024a}, where the loss must respect the periodicity of the unknown phases~\cite{Demkowicz-Dobrzanski2011,Holevo2011b}.

For phase parameters $\boldsymbol{\theta}\in\mathbb T^d$, where $\mathbb T=\mathbb R/(2\pi\mathbb Z)$ denotes the circle, consider the additive loss
\begin{equation}
    \ell(\boldsymbol{\tilde{\theta}},\boldsymbol{\theta})
    =
    \sum_{j=1}^d w_j g_j(\tilde{\theta}_j-\theta_j),
    \label{eq:holevo_loss}
\end{equation}
where each $g_j(z)$ is nonnegative, even, continuous, and $2\pi$-periodic in $z$.
In covariant estimation~\cite{Macchiavello2003, Demkowicz-Dobrzanski2015a}, an important class of such losses is defined by Holevo's condition~\cite[Thm.~4.4.1]{Holevo2011b}.
This condition requires the nonconstant Fourier coefficients of each $g_j$ to be nonpositive.
Explicitly, defining
\begin{equation}
    \hat g_j(k)
    =
    \frac{1}{2\pi}\int_0^{2\pi}\dd z\,g_j(z)e^{-ikz},
    \quad k\in\mathbb Z,
\end{equation}
the condition reads
\begin{equation}
    \hat g_j(k)
    =
    \frac{1}{2\pi}\int_0^{2\pi}\dd z\,g_j(z)\cos(kz)
    \leq0,\quad k\neq0.
    \label{eq:holevo_condition}
\end{equation}
We have used that $g_j$ is real and even, so the sine terms in the integral vanish, leaving only the cosine terms, with real Fourier coefficients satisfying $\hat g_j(-k)=\hat g_j(k)$.
The constant coefficient $\hat g_j(0)$, which gives the average of $g_j$ over one period, is not constrained by this condition.

Eq.~\eqref{eq:holevo_condition} also guarantees the positivity required by our incompatibility bound.
To prove this, choose $c_j>0$ such that $c_j\geq\max_z g_j(z)$ and define the score functions
\begin{equation}
    \sigma_j(\tilde{\theta}_j,\theta_j)
    =1-\frac{g_j(\tilde{\theta}_j-\theta_j)}{c_j}.
    \label{eq:holevo_score}
\end{equation}
As functions of $\tilde{\theta}_j-\theta_j$, these scores have Fourier coefficients
\begin{equation}
    a_{j,k}=\delta_{k0}-\frac{\hat g_j(k)}{c_j}\geq0.
\end{equation}
For $k\neq0$, their nonnegativity follows from Eq.~\eqref{eq:holevo_condition}; for $k=0$, it follows by construction because the average $\hat g_j(0)$ cannot exceed $\max_z g_j(z)$.

The required self-convolution property can be shown explicitly by introducing the following real-valued functions:
\begin{equation}
\begin{aligned}
    h_j(\varphi,0)&=\sqrt{a_{j,0}},\\
    h_j(\varphi,2k-1)&=\sqrt{2a_{j,k}}\cos(k\varphi),\\
    h_j(\varphi,2k)&=\sqrt{2a_{j,k}}\sin(k\varphi),
    \quad k\geq1.
\end{aligned}
\end{equation}
The trigonometric identity 
\begin{equation}
    \cos[k(\varphi-\psi)]=\cos(k\varphi)\cos(k\psi)+\sin(k\varphi)\sin(k\psi)
\end{equation}
gives
\begin{equation}
\begin{aligned}
    \sigma_j(\varphi,\psi)
    &=a_{j,0}+2\sum_{k\geq1}a_{j,k}\cos[k(\varphi-\psi)]\\
    &=\sum_{\alpha=0}^{\infty}h_j(\varphi,\alpha)h_j(\psi,\alpha)
    =(h_j\star h_j)(\varphi,\psi).
\end{aligned}
\end{equation}
Here, the variable $t$ in Eq.~\eqref{eq:sigma_score_def} becomes the index $\alpha=0,1,2,\ldots$ labelling the $h$-functions above, while the measure assigns unit weight to each value, so $\int \dd t\,\mu(t)$ becomes $\sum_{\alpha=0}^{\infty}$.
Continuity and nonnegative Fourier coefficients ensure that 
\begin{equation}
    \sum_{k\in\mathbb Z}a_{j,k}=\sigma_j(\varphi,\varphi)\leq1,
\end{equation}
so the sums converge and the functions $h_j$ are bounded.
Thus, each $\sigma_j$ is a positive score.

Consequently, Eq.~\eqref{eq:factor_two} applies to additive Holevo losses.
This is a key metrological result of this work.
Importantly, this conclusion holds for arbitrary priors, even though covariant phase measurements are typically associated with uniform priors over $[0,2\pi)$ for each phase, corresponding to maximum ignorance.

As a particular example, consider the sine loss~\cite{Demkowicz-Dobrzanski2011,Macchiavello2003}
\begin{equation}
    \ell(\boldsymbol{\tilde{\theta}},\boldsymbol{\theta})
    =
    4\sum_{j=1}^d w_j
    \sin^2\left(\frac{\tilde{\theta}_j-\theta_j}{2}\right).
    \label{eq:additive_ese}
\end{equation}
This loss is sometimes referred to as the Holevo cost function, and the corresponding mean loss as the Collett phase variance (CPV)~\cite{Collett1993,Berry2009}, which we denote by $\mathcal L_{\mathrm{CPV}}$.

We see that $g_j(z)=4\sin^2(z/2)$ satisfies Eq.~\eqref{eq:holevo_condition}, and choosing $c_j=4$ gives the positive score
\begin{equation}
    \sigma_j(\tilde{\theta}_j,\theta_j)
    =
    \cos^2\left(\frac{\tilde{\theta}_j-\theta_j}{2}\right).
    \label{eq:sine_positive_score}
\end{equation}
Thus, Eq.~\eqref{eq:factor_two} applies to the CPV, which, denoting by $\mathcal L_{\mathrm{CPV},\IF}$ its incompatibility-free benchmark, can be written as
\begin{equation}
    \mathcal L_{\mathrm{CPV},\IF}
    \leq
    \mathcal L_{\mathrm{CPV},\opt}
    \leq
    \mathcal L_{\mathrm{CPV}}(M_{\PGM})
    \leq
    2\mathcal L_{\mathrm{CPV},\IF}.
\end{equation}

For this loss, $\mathcal L_{\mathrm{CPV},\IF}$ can be evaluated explicitly, as in the case of quadratic loss~\cite{Personick1971,Rubio2020a}, using the single-phase result of Ref.~\cite{Demkowicz-Dobrzanski2011}.
Defining the first circular moment operators
\begin{equation}
    \bar{\rho}_j
    =
    \int_{\mathbb T^d}\dd\boldsymbol{\theta}\,
    p(\boldsymbol{\theta})
    \rho(\boldsymbol{\theta})e^{-i\theta_j},
    \label{eq:first_circular_moment}
\end{equation}
the minimum CPV for estimating each parameter separately is
\begin{equation}
    \mathcal L_{\mathrm{CPV},j}^*
    =
    2(1-\|\bar{\rho}_j\|_1),
\end{equation}
where $\|\cdot\|_1$ denotes the trace norm.
Consequently,
\begin{equation}
    \mathcal L_{\mathrm{CPV},\IF}
    =
    2\sum_{j=1}^d w_j(1-\|\bar{\rho}_j\|_1).
    \label{eq:ese_ni_explicit}
\end{equation}

\section{Extension to unbounded losses}
\label{sec:unbounded}

While fairly general, losses of the form in Eq.~\eqref{eq:score_loss} are bounded by $c_j$, whereas quadratic losses can grow arbitrarily large.
Hence, we next extend the argument above to unbounded losses by introducing positive exponential scores.
The following theorem follows from a straightforward extension of a construction introduced in Ref.~\cite{Mishra2025} to deal with quadratic loss.

\begin{theorem}
\label{th:unbounded}
Suppose that, for every $j=1,\ldots,d$ and every $\kappa>0$,
\begin{equation}
    \sigma_{\kappa,j}(\tilde{\theta}_j,\theta_j)
    =
    \exp[
        -\kappa\ell_j(\tilde{\theta}_j,\theta_j)
    ]
    \label{eq:exponential_score}
\end{equation}
is a positive score on $\Omega_j$. 
Then the bound in Eq.~\eqref{eq:factor_two} also holds for this additive loss.
The statement is nontrivial when $\mathcal L_{\IF}<\infty$; otherwise the upper bound is void.
\end{theorem}

\begin{proof}
As before, it suffices to prove
$\mathcal L_j(M_{\PGM})\leq2\mathcal L_j^*$ for each $j$.
For later convenience, define
\begin{equation}
    q_\kappa(x)
    \coloneq
    \frac{1-\exp(-\kappa x)}{\kappa},
    \quad
    \kappa>0,
\end{equation}
which satisfies
\begin{equation}
    q_\kappa(x)\uparrow x
    \quad
    (\kappa\downarrow0)
\end{equation}
for every $x\geq0$.

Fix a component $j$, a value $\kappa>0$, and an arbitrary POVM $M_j$ with outcomes in $\Omega_j$.
Extend $M_j$ to a POVM on $\Omega$ and consider the positive score
\begin{equation}
    \hat{\sigma}_{\kappa,j}
    (\boldsymbol{\tilde{\theta}},\boldsymbol{\theta})
    =
    \exp[
        -\kappa
        \ell_j(\tilde{\theta}_j,\theta_j)
    ].
\end{equation}

Let $\varsigma_{\kappa,j}(M_{\PGM})$ and $\varsigma_{\kappa,j}(M_j)$ denote the corresponding expected scores of $M_{\PGM}$ and $M_j$, respectively.
We have that
\begin{equation}
    \varsigma_{\kappa,j}(M_j)
    \leq
    \varsigma_{\kappa,j}^*.
\end{equation}
Applying the generalised BK inequality then gives
\begin{equation}
    \varsigma_{\kappa,j}^*
    \leq
    \sqrt{\varsigma_{\kappa,j}(M_{\PGM})},
\end{equation}
and hence
\begin{equation}
    \varsigma_{\kappa,j}(M_j)
    \leq
    \sqrt{\varsigma_{\kappa,j}(M_{\PGM})}.
\end{equation}
Therefore,
\begin{align}
    1-\varsigma_{\kappa,j}(M_{\PGM})
    &\leq
    1-\varsigma_{\kappa,j}(M_j)^2
    \nonumber\\
    &=
    [1-\varsigma_{\kappa,j}(M_j)]
    [1+\varsigma_{\kappa,j}(M_j)]
    \nonumber\\
    &\leq
    2[1-\varsigma_{\kappa,j}(M_j)].
    \label{eq:deficit_bound}
\end{align}

Dividing Eq.~\eqref{eq:deficit_bound} by $\kappa$ and using the definition of $q_\kappa$ gives
\begin{align}
&
    \int\dd\boldsymbol{\theta}
    \dd\boldsymbol{\tilde{\theta}}\,
    p(\boldsymbol{\theta})
    \Tr[
        \rho(\boldsymbol{\theta})
        M_{\PGM}(\boldsymbol{\tilde{\theta}})
    ]
    q_\kappa[
        \ell_j(\tilde{\theta}_j,\theta_j)
    ]
    \nonumber\\
&\quad\leq
    2
    \int\dd\boldsymbol{\theta}
    \dd\tilde{\theta}_j\,
    p(\boldsymbol{\theta})
    \Tr[
        \rho(\boldsymbol{\theta})
        M_j(\tilde{\theta}_j)
    ]
    q_\kappa[
        \ell_j(\tilde{\theta}_j,\theta_j)
    ].
\end{align}
Taking $\kappa\downarrow0$ and applying the monotone convergence theorem yields
\begin{equation}
    \mathcal L_j(M_{\PGM})
    \leq
    2\mathcal L_j(M_j).
\end{equation}
Minimising over $M_j$ gives
\begin{equation}
    \mathcal L_j(M_{\PGM})
    \leq
    2\mathcal L_j^*.
\end{equation}
Multiplying by $w_j$ and summing over $j$ proves
\begin{equation}
    \mathcal L(M_{\PGM})
    \leq
    2\mathcal L_{\IF},
\end{equation}
while the remaining inequalities in Eq.~\eqref{eq:factor_two} follow as before.
\end{proof}
We note that Thm.~\ref{th:unbounded} actually includes Thm.~\ref{thm:factor_two} as a special case: if $\sigma_j$ is a positive score, the exponential score $\exp[-\kappa c_j(1-\sigma_j)]$ is also positive for every $\kappa>0$, since its Taylor expansion in $\sigma_j$ has nonnegative coefficients and products of positive scores remain positive.
We have kept the two results separate for clarity, as the proof of Thm.~\ref{thm:factor_two} is more direct.
Moreover, for bounded losses the intermediate bound $\mathcal L_j(M_{\PGM})\leq2\mathcal L_j^*-(\mathcal L_j^*)^2/c_j$, obtained from the first inequality in Eq.~\eqref{eq:ineq_chain_PGM_factortwo}, is slightly tighter than $2\mathcal L_j^*$.

\begin{remark}[Negative-type losses]
A sufficient and necessary condition for Eq.~\eqref{eq:exponential_score} to be a positive score on $\Omega_j$ is that each $\ell_j$ is symmetric and conditionally negative definite in the sense of Schoenberg~\cite{Berg1984}.
In this case, $\exp(-\kappa\ell_j)$ is positive definite for every $\kappa>0$ and hence defines a positive score.
Squared Euclidean distances, and more generally squared Hilbert-space distances, provide standard examples.
\end{remark}

For $\ell_j(\tilde{\theta}_j,\theta_j)=(\tilde{\theta}_j-\theta_j)^2$, the exponential scores are positive for every $\kappa>0$.
Thus, Thm.~\ref{th:unbounded} recovers the result of Ref.~\cite{Albarelli2026c} for quadratic losses, in the form of Eq.~\eqref{eq:factor_two}.
This also extends to a nondiagonal weight matrix $W\succeq0$ through its spectral decomposition $W=\sum_j w_j\boldsymbol u_j\boldsymbol u_j^{\mathsf T}$, since the loss is diagonal in the rotated parameters $\boldsymbol u_j^{\mathsf T}\boldsymbol{\theta}$.

Finally, as an example of a nonquadratic unbounded loss, consider the additive absolute loss
\begin{equation}
    \ell(\boldsymbol{\tilde{\theta}},\boldsymbol{\theta})
    =
    \sum_{j=1}^d w_j
    |\tilde{\theta}_j-\theta_j|.
\end{equation}
Its exponential scores are the Laplace kernels
\begin{equation}
    \sigma_{\kappa,j}(\tilde{\theta}_j,\theta_j)
    =
    e^{-\kappa|\tilde{\theta}_j-\theta_j|},
\end{equation}
which are positive definite, as follows from their nonnegative Fourier transforms. 
The bound in Eq.~\eqref{eq:factor_two} therefore follows from Thm.~\ref{th:unbounded}.
Absolute losses, whose associated optimal Bayesian estimator is the posterior median, are motivated by their robustness to outliers, as their linear penalty avoids assigning disproportionate weight to large deviations~\cite{Jaynes2003a}.

\section{Concluding remarks}
\label{sec:conclusions}

We have shown that the factor-of-two limitation on measurement incompatibility in Bayesian multiparameter estimation is not specific to quadratic loss, but can follow more generally from the positivity of the score functions associated with additive losses.
Crucially, this establishes the corresponding incompatibility bound for Holevo's family of periodic losses in Bayesian multiphase estimation with arbitrary prior information, despite the absence of a globally equivalent quadratic description.
The same argument covers unbounded losses, such as absolute and quadratic losses.

Our benchmark $\mathcal L_{\IF}$ is obtained from single-parameter estimation problems in which the remaining components of $\boldsymbol{\theta}$ act as nuisance parameters.
This parallels using the Helstrom SLD-CRB as a benchmark for incompatibility in local QET~\cite{Belliardo2021}, since, for a diagonal weight matrix, it is the sum of the nuisance-parameter bounds $(F_Q^{-1})_{jj}$, where $F_Q$ is the SLD quantum Fisher information matrix~\cite{Suzuki2020b}.
An alternative notion of incompatibility takes as the benchmark ideal single-parameter experiments in which all the other parameters are perfectly known.
It was introduced in Ref.~\cite{Ragy2016} for local QET (also including probe selection for channel estimation problems), and in that setting it amounts to the requirement that the single-parameter CRBs $1/(F_Q)_{jj}$ are attained simultaneously for all $\theta_j$.
Thus, a compatible model is not just one without measurement incompatibility, but one where $F_Q$ is also diagonal, so that the uncertainty about one parameter does not affect the estimation of the others.
A Bayesian counterpart of this concept can also be defined~\footnote{R.~Demkowicz-Dobrza\'nski, private communication; the argument will appear in a forthcoming book.};
compatibility in this sense is a stronger requirement, which implies the absence of measurement incompatibility according to our definition, while the converse does not hold.
Although this alternative notion may provide a unified view of precision trade-offs in multiparameter estimation, it does not isolate measurement incompatibility per se, since the additional cost of not knowing the other parameters is already present in classical models.

As an outlook, it would be interesting to understand whether the impact of measurement incompatibility can be bounded even beyond the class of positive scores required by the generalised BK inequality; this will likely require going beyond the PGM-based approach used here~\cite{Grewal2026a}.
Answering this question would be relevant, for example, for extending to multiple parameters the learning framework of Refs.~\cite{Meyer2025,Belliardo2026b}, which relies on a score that assigns the value one if the estimate is within a prescribed tolerance of the true value and zero otherwise.
From a more general point of view, our work provides further evidence of a fruitful interplay between parameter estimation and hypothesis testing.

\begin{acknowledgments}
We gratefully acknowledge insightful discussions with Dominic Branford, Rafał Demkowicz-Dobrzański, and Ludovico Lami.
SC was supported by the PRIN project
PRIN22-RISQUE-2022T25TR3 of the Italian Ministry of University and Research (MUR).
JR acknowledges financial support from the UKRI-JST project ``Quantum Control \& Sensing: Enhancing high-precision quantum sensing in noisy environments via optimised control''.

\bigskip

\emph{ChatGPT (OpenAI) and Claude (Anthropic) were used in conversational mode during the preparation of this manuscript.
The authors independently developed and verified the technical results and take full responsibility for the final content of the manuscript.}

\vfill

\end{acknowledgments}

\bibliography{2026BayesianMulti}

\begin{thebibliography}{56}%
\makeatletter
\providecommand \@ifxundefined [1]{%
 \@ifx{#1\undefined}
}%
\providecommand \@ifnum [1]{%
 \ifnum #1\expandafter \@firstoftwo
 \else \expandafter \@secondoftwo
 \fi
}%
\providecommand \@ifx [1]{%
 \ifx #1\expandafter \@firstoftwo
 \else \expandafter \@secondoftwo
 \fi
}%
\providecommand \natexlab [1]{#1}%
\providecommand \enquote  [1]{``#1''}%
\providecommand \bibnamefont  [1]{#1}%
\providecommand \bibfnamefont [1]{#1}%
\providecommand \citenamefont [1]{#1}%
\providecommand \href@noop [0]{\@secondoftwo}%
\providecommand \href [0]{\begingroup \@sanitize@url \@href}%
\providecommand \@href[1]{\@@startlink{#1}\@@href}%
\providecommand \@@href[1]{\endgroup#1\@@endlink}%
\providecommand \@sanitize@url [0]{\catcode `\\12\catcode `\$12\catcode `\&12\catcode `\#12\catcode `\^12\catcode `\_12\catcode `\%12\relax}%
\providecommand \@@startlink[1]{}%
\providecommand \@@endlink[0]{}%
\providecommand \url  [0]{\begingroup\@sanitize@url \@url }%
\providecommand \@url [1]{\endgroup\@href {#1}{\urlprefix }}%
\providecommand \urlprefix  [0]{URL }%
\providecommand \Eprint [0]{\href }%
\providecommand \doibase [0]{https://doi.org/}%
\providecommand \selectlanguage [0]{\@gobble}%
\providecommand \bibinfo  [0]{\@secondoftwo}%
\providecommand \bibfield  [0]{\@secondoftwo}%
\providecommand \translation [1]{[#1]}%
\providecommand \BibitemOpen [0]{}%
\providecommand \bibitemStop [0]{}%
\providecommand \bibitemNoStop [0]{.\EOS\space}%
\providecommand \EOS [0]{\spacefactor3000\relax}%
\providecommand \BibitemShut  [1]{\csname bibitem#1\endcsname}%
\let\auto@bib@innerbib\@empty
\bibitem [{\citenamefont {Albarelli}\ \emph {et~al.}(2020)\citenamefont {Albarelli}, \citenamefont {Barbieri}, \citenamefont {Genoni},\ and\ \citenamefont {Gianani}}]{Albarelli2020a}%
  \BibitemOpen
  \bibfield  {author} {\bibinfo {author} {\bibfnamefont {F.}~\bibnamefont {Albarelli}}, \bibinfo {author} {\bibfnamefont {M.}~\bibnamefont {Barbieri}}, \bibinfo {author} {\bibfnamefont {M.~G.}\ \bibnamefont {Genoni}},\ and\ \bibinfo {author} {\bibfnamefont {I.}~\bibnamefont {Gianani}},\ }\bibfield  {title} {\bibinfo {title} {A perspective on multiparameter quantum metrology: {{From}} theoretical tools to applications in quantum imaging},\ }\href {https://doi.org/10.1016/j.physleta.2020.126311} {\bibfield  {journal} {\bibinfo  {journal} {Phys. Lett. A}\ }\textbf {\bibinfo {volume} {384}},\ \bibinfo {pages} {126311} (\bibinfo {year} {2020})},\ \Eprint {https://arxiv.org/abs/1911.12067} {arXiv:1911.12067} \BibitemShut {NoStop}%
\bibitem [{\citenamefont {{Demkowicz-Dobrza{\'n}ski}}\ \emph {et~al.}(2020)\citenamefont {{Demkowicz-Dobrza{\'n}ski}}, \citenamefont {G{\'o}recki},\ and\ \citenamefont {Gu{\c t}{\u a}}}]{Demkowicz-Dobrzanski2020}%
  \BibitemOpen
  \bibfield  {author} {\bibinfo {author} {\bibfnamefont {R.}~\bibnamefont {{Demkowicz-Dobrza{\'n}ski}}}, \bibinfo {author} {\bibfnamefont {W.}~\bibnamefont {G{\'o}recki}},\ and\ \bibinfo {author} {\bibfnamefont {M.}~\bibnamefont {Gu{\c t}{\u a}}},\ }\bibfield  {title} {\bibinfo {title} {Multi-parameter estimation beyond quantum {{Fisher}} information},\ }\href {https://doi.org/10.1088/1751-8121/ab8ef3} {\bibfield  {journal} {\bibinfo  {journal} {J. Phys. A}\ }\textbf {\bibinfo {volume} {53}},\ \bibinfo {pages} {363001} (\bibinfo {year} {2020})},\ \Eprint {https://arxiv.org/abs/2001.11742} {arXiv:2001.11742} \BibitemShut {NoStop}%
\bibitem [{\citenamefont {Pezz{\`e}}\ and\ \citenamefont {Smerzi}(2025)}]{Pezze2025}%
  \BibitemOpen
  \bibfield  {author} {\bibinfo {author} {\bibfnamefont {L.}~\bibnamefont {Pezz{\`e}}}\ and\ \bibinfo {author} {\bibfnamefont {A.}~\bibnamefont {Smerzi}},\ }\href {https://doi.org/10.48550/arXiv.2502.17396} {\bibinfo {title} {Advances in multiparameter quantum sensing and metrology}} (\bibinfo {year} {2025}),\ \Eprint {https://arxiv.org/abs/2502.17396} {arXiv:2502.17396} \BibitemShut {NoStop}%
\bibitem [{\citenamefont {Conlon}\ \emph {et~al.}(2026)\citenamefont {Conlon}, \citenamefont {Shajilal}, \citenamefont {Zhao}, \citenamefont {Ralph}, \citenamefont {Leuchs}, \citenamefont {Andersen}, \citenamefont {Assad}, \citenamefont {Lam},\ and\ \citenamefont {Xiao}}]{Conlon2026a}%
  \BibitemOpen
  \bibfield  {author} {\bibinfo {author} {\bibfnamefont {L.~O.}\ \bibnamefont {Conlon}}, \bibinfo {author} {\bibfnamefont {B.}~\bibnamefont {Shajilal}}, \bibinfo {author} {\bibfnamefont {J.}~\bibnamefont {Zhao}}, \bibinfo {author} {\bibfnamefont {T.~C.}\ \bibnamefont {Ralph}}, \bibinfo {author} {\bibfnamefont {G.}~\bibnamefont {Leuchs}}, \bibinfo {author} {\bibfnamefont {U.~L.}\ \bibnamefont {Andersen}}, \bibinfo {author} {\bibfnamefont {S.~M.}\ \bibnamefont {Assad}}, \bibinfo {author} {\bibfnamefont {P.~K.}\ \bibnamefont {Lam}},\ and\ \bibinfo {author} {\bibfnamefont {Y.}~\bibnamefont {Xiao}},\ }\href {https://doi.org/10.48550/arXiv.2606.07747} {\bibinfo {title} {\$100\textbackslash{{pm$\Delta$t}}\$ {{Years}} of {{Quantum Uncertainty}}: {{From Origins}} to {{Modern Insights}}}} (\bibinfo {year} {2026}),\ \Eprint {https://arxiv.org/abs/2606.07747} {arXiv:2606.07747} \BibitemShut {NoStop}%
\bibitem [{\citenamefont {Audenaert}\ and\ \citenamefont {Mosonyi}(2014)}]{Audenaert2014a}%
  \BibitemOpen
  \bibfield  {author} {\bibinfo {author} {\bibfnamefont {K.~M.~R.}\ \bibnamefont {Audenaert}}\ and\ \bibinfo {author} {\bibfnamefont {M.}~\bibnamefont {Mosonyi}},\ }\bibfield  {title} {\bibinfo {title} {Upper bounds on the error probabilities and asymptotic error exponents in quantum multiple state discrimination},\ }\href {https://doi.org/10.1063/1.4898559} {\bibfield  {journal} {\bibinfo  {journal} {J. Math. Phys.}\ }\textbf {\bibinfo {volume} {55}},\ \bibinfo {pages} {102201} (\bibinfo {year} {2014})},\ \Eprint {https://arxiv.org/abs/1401.7658} {arXiv:1401.7658} \BibitemShut {NoStop}%
\bibitem [{\citenamefont {Li}(2016)}]{liDiscriminatingQuantumStates2016}%
  \BibitemOpen
  \bibfield  {author} {\bibinfo {author} {\bibfnamefont {K.}~\bibnamefont {Li}},\ }\bibfield  {title} {\bibinfo {title} {Discriminating quantum states: {{The}} multiple {{Chernoff}} distance},\ }\href {https://doi.org/10.1214/16-AOS1436} {\bibfield  {journal} {\bibinfo  {journal} {Ann. Stat.}\ }\textbf {\bibinfo {volume} {44}},\ \bibinfo {pages} {1661} (\bibinfo {year} {2016})},\ \Eprint {https://arxiv.org/abs/1508.06624} {arXiv:1508.06624} \BibitemShut {NoStop}%
\bibitem [{\citenamefont {Cheng}\ and\ \citenamefont {Liu}(2026)}]{Cheng2026}%
  \BibitemOpen
  \bibfield  {author} {\bibinfo {author} {\bibfnamefont {H.-C.}\ \bibnamefont {Cheng}}\ and\ \bibinfo {author} {\bibfnamefont {P.-C.}\ \bibnamefont {Liu}},\ }\href {https://doi.org/10.48550/arXiv.2606.06246} {\bibinfo {title} {Multiple {{Quantum Hypothesis Testing}}: {{One-Shot Pairwise Bounds}} and {{Sharp Asymptotics}}}} (\bibinfo {year} {2026}),\ \Eprint {https://arxiv.org/abs/2606.06246} {arXiv:2606.06246} \BibitemShut {NoStop}%
\bibitem [{\citenamefont {Lee}\ and\ \citenamefont {Lami}(2026)}]{Lee2026b}%
  \BibitemOpen
  \bibfield  {author} {\bibinfo {author} {\bibfnamefont {K.-Y.}\ \bibnamefont {Lee}}\ and\ \bibinfo {author} {\bibfnamefont {L.}~\bibnamefont {Lami}},\ }\href {https://doi.org/10.48550/arXiv.2609.28440} {\bibinfo {title} {Sharp pairwise reduction for quantum hypothesis testing}} (\bibinfo {year} {2026}),\ \Eprint {https://arxiv.org/abs/2609.28440} {arXiv:2609.28440} \BibitemShut {NoStop}%
\bibitem [{\citenamefont {Belliardo}\ and\ \citenamefont {Giovannetti}(2021)}]{Belliardo2021}%
  \BibitemOpen
  \bibfield  {author} {\bibinfo {author} {\bibfnamefont {F.}~\bibnamefont {Belliardo}}\ and\ \bibinfo {author} {\bibfnamefont {V.}~\bibnamefont {Giovannetti}},\ }\bibfield  {title} {\bibinfo {title} {Incompatibility in quantum parameter estimation},\ }\href {https://doi.org/10.1088/1367-2630/ac04ca} {\bibfield  {journal} {\bibinfo  {journal} {New J. Phys.}\ }\textbf {\bibinfo {volume} {23}},\ \bibinfo {pages} {063055} (\bibinfo {year} {2021})},\ \Eprint {https://arxiv.org/abs/2102.13417} {arXiv:2102.13417} \BibitemShut {NoStop}%
\bibitem [{\citenamefont {Sidhu}\ \emph {et~al.}(2021)\citenamefont {Sidhu}, \citenamefont {Ouyang}, \citenamefont {Campbell},\ and\ \citenamefont {Kok}}]{Sidhu2021}%
  \BibitemOpen
  \bibfield  {author} {\bibinfo {author} {\bibfnamefont {J.~S.}\ \bibnamefont {Sidhu}}, \bibinfo {author} {\bibfnamefont {Y.}~\bibnamefont {Ouyang}}, \bibinfo {author} {\bibfnamefont {E.~T.}\ \bibnamefont {Campbell}},\ and\ \bibinfo {author} {\bibfnamefont {P.}~\bibnamefont {Kok}},\ }\bibfield  {title} {\bibinfo {title} {Tight {{Bounds}} on the {{Simultaneous Estimation}} of {{Incompatible Parameters}}},\ }\href {https://doi.org/10.1103/PhysRevX.11.011028} {\bibfield  {journal} {\bibinfo  {journal} {Phys. Rev. X}\ }\textbf {\bibinfo {volume} {11}},\ \bibinfo {pages} {011028} (\bibinfo {year} {2021})},\ \Eprint {https://arxiv.org/abs/1912.09218} {arXiv:1912.09218} \BibitemShut {NoStop}%
\bibitem [{\citenamefont {Candeloro}\ \emph {et~al.}(2021)\citenamefont {Candeloro}, \citenamefont {Paris},\ and\ \citenamefont {Genoni}}]{Candeloro2021c}%
  \BibitemOpen
  \bibfield  {author} {\bibinfo {author} {\bibfnamefont {A.}~\bibnamefont {Candeloro}}, \bibinfo {author} {\bibfnamefont {M.~G.~A.}\ \bibnamefont {Paris}},\ and\ \bibinfo {author} {\bibfnamefont {M.~G.}\ \bibnamefont {Genoni}},\ }\bibfield  {title} {\bibinfo {title} {On the properties of the asymptotic incompatibility measure in multiparameter quantum estimation},\ }\href {https://doi.org/10.1088/1751-8121/ac331e} {\bibfield  {journal} {\bibinfo  {journal} {J. Phys. A}\ }\textbf {\bibinfo {volume} {54}},\ \bibinfo {pages} {485301} (\bibinfo {year} {2021})},\ \Eprint {https://arxiv.org/abs/2107.13426} {arXiv:2107.13426} \BibitemShut {NoStop}%
\bibitem [{\citenamefont {Chen}\ \emph {et~al.}(2022)\citenamefont {Chen}, \citenamefont {Chen},\ and\ \citenamefont {Yuan}}]{Chen2022f}%
  \BibitemOpen
  \bibfield  {author} {\bibinfo {author} {\bibfnamefont {H.}~\bibnamefont {Chen}}, \bibinfo {author} {\bibfnamefont {Y.}~\bibnamefont {Chen}},\ and\ \bibinfo {author} {\bibfnamefont {H.}~\bibnamefont {Yuan}},\ }\bibfield  {title} {\bibinfo {title} {Incompatibility measures in multiparameter quantum estimation under hierarchical quantum measurements},\ }\href {https://doi.org/10.1103/PhysRevA.105.062442} {\bibfield  {journal} {\bibinfo  {journal} {Phys. Rev. A}\ }\textbf {\bibinfo {volume} {105}},\ \bibinfo {pages} {062442} (\bibinfo {year} {2022})},\ \Eprint {https://arxiv.org/abs/2109.05807} {arXiv:2109.05807} \BibitemShut {NoStop}%
\bibitem [{\citenamefont {Hayashi}\ and\ \citenamefont {Ouyang}(2023)}]{Hayashi2023a}%
  \BibitemOpen
  \bibfield  {author} {\bibinfo {author} {\bibfnamefont {M.}~\bibnamefont {Hayashi}}\ and\ \bibinfo {author} {\bibfnamefont {Y.}~\bibnamefont {Ouyang}},\ }\bibfield  {title} {\bibinfo {title} {Tight {{Cram\'er-Rao}} type bounds for multiparameter quantum metrology through conic programming},\ }\href {https://doi.org/10.22331/q-2023-08-29-1094} {\bibfield  {journal} {\bibinfo  {journal} {Quantum}\ }\textbf {\bibinfo {volume} {7}},\ \bibinfo {pages} {1094} (\bibinfo {year} {2023})},\ \Eprint {https://arxiv.org/abs/2209.05218} {arXiv:2209.05218} \BibitemShut {NoStop}%
\bibitem [{\citenamefont {Candeloro}\ \emph {et~al.}(2024)\citenamefont {Candeloro}, \citenamefont {Pazhotan},\ and\ \citenamefont {Paris}}]{Candeloro2024}%
  \BibitemOpen
  \bibfield  {author} {\bibinfo {author} {\bibfnamefont {A.}~\bibnamefont {Candeloro}}, \bibinfo {author} {\bibfnamefont {Z.}~\bibnamefont {Pazhotan}},\ and\ \bibinfo {author} {\bibfnamefont {M.~G.~A.}\ \bibnamefont {Paris}},\ }\bibfield  {title} {\bibinfo {title} {Dimension matters: Precision and incompatibility in multi-parameter quantum estimation models},\ }\href {https://doi.org/10.1088/2058-9565/ad7498} {\bibfield  {journal} {\bibinfo  {journal} {Quantum Sci. Technol.}\ }\textbf {\bibinfo {volume} {9}},\ \bibinfo {pages} {045045} (\bibinfo {year} {2024})},\ \Eprint {https://arxiv.org/abs/2403.07106} {arXiv:2403.07106} \BibitemShut {NoStop}%
\bibitem [{\citenamefont {Chen}\ \emph {et~al.}(2024)\citenamefont {Chen}, \citenamefont {Wang},\ and\ \citenamefont {Yuan}}]{Chen2024j}%
  \BibitemOpen
  \bibfield  {author} {\bibinfo {author} {\bibfnamefont {H.}~\bibnamefont {Chen}}, \bibinfo {author} {\bibfnamefont {L.}~\bibnamefont {Wang}},\ and\ \bibinfo {author} {\bibfnamefont {H.}~\bibnamefont {Yuan}},\ }\bibfield  {title} {\bibinfo {title} {Simultaneous measurement of multiple incompatible observables and tradeoff in multiparameter quantum estimation},\ }\href {https://doi.org/10.1038/s41534-024-00894-x} {\bibfield  {journal} {\bibinfo  {journal} {Npj Quantum Inf.}\ }\textbf {\bibinfo {volume} {10}},\ \bibinfo {pages} {98} (\bibinfo {year} {2024})},\ \Eprint {https://arxiv.org/abs/2310.11925} {arXiv:2310.11925} \BibitemShut {NoStop}%
\bibitem [{\citenamefont {Das}\ \emph {et~al.}(2025)\citenamefont {Das}, \citenamefont {Conlon}, \citenamefont {Suzuki}, \citenamefont {Yung}, \citenamefont {Lam},\ and\ \citenamefont {Assad}}]{Das2025b}%
  \BibitemOpen
  \bibfield  {author} {\bibinfo {author} {\bibfnamefont {A.}~\bibnamefont {Das}}, \bibinfo {author} {\bibfnamefont {L.~O.}\ \bibnamefont {Conlon}}, \bibinfo {author} {\bibfnamefont {J.}~\bibnamefont {Suzuki}}, \bibinfo {author} {\bibfnamefont {S.~K.}\ \bibnamefont {Yung}}, \bibinfo {author} {\bibfnamefont {P.~K.}\ \bibnamefont {Lam}},\ and\ \bibinfo {author} {\bibfnamefont {S.~M.}\ \bibnamefont {Assad}},\ }\bibfield  {title} {\bibinfo {title} {Holevo {{Cram\'er-Rao}} bound: {{How}} close can we get without entangling measurements?},\ }\href {https://doi.org/10.22331/q-2025-09-30-1867} {\bibfield  {journal} {\bibinfo  {journal} {Quantum}\ }\textbf {\bibinfo {volume} {9}},\ \bibinfo {pages} {1867} (\bibinfo {year} {2025})},\ \Eprint {https://arxiv.org/abs/2405.09622} {arXiv:2405.09622} \BibitemShut {NoStop}%
\bibitem [{\citenamefont {He}\ \emph {et~al.}(2026)\citenamefont {He}, \citenamefont {Fazio},\ and\ \citenamefont {Paris}}]{He2026d}%
  \BibitemOpen
  \bibfield  {author} {\bibinfo {author} {\bibfnamefont {J.}~\bibnamefont {He}}, \bibinfo {author} {\bibfnamefont {G.}~\bibnamefont {Fazio}},\ and\ \bibinfo {author} {\bibfnamefont {M.~G.~A.}\ \bibnamefont {Paris}},\ }\href {https://doi.org/10.48550/arXiv.2510.18864} {\bibinfo {title} {Weight-dependent and weight-independent measures of quantum incompatibility in multiparameter estimation}} (\bibinfo {year} {2026}),\ \Eprint {https://arxiv.org/abs/2510.18864} {arXiv:2510.18864} \BibitemShut {NoStop}%
\bibitem [{\citenamefont {He}\ and\ \citenamefont {Paris}(2026)}]{He2026c}%
  \BibitemOpen
  \bibfield  {author} {\bibinfo {author} {\bibfnamefont {J.}~\bibnamefont {He}}\ and\ \bibinfo {author} {\bibfnamefont {M.~G.~A.}\ \bibnamefont {Paris}},\ }\href {https://doi.org/10.48550/arXiv.2606.11343} {\bibinfo {title} {Fisher geometry reshapes the effect of incompatibility in multiparameter quantum estimation}} (\bibinfo {year} {2026}),\ \Eprint {https://arxiv.org/abs/2606.11343} {arXiv:2606.11343} \BibitemShut {NoStop}%
\bibitem [{\citenamefont {Wang}\ \emph {et~al.}(2026)\citenamefont {Wang}, \citenamefont {Chen},\ and\ \citenamefont {Yuan}}]{Wang2026n}%
  \BibitemOpen
  \bibfield  {author} {\bibinfo {author} {\bibfnamefont {L.}~\bibnamefont {Wang}}, \bibinfo {author} {\bibfnamefont {H.}~\bibnamefont {Chen}},\ and\ \bibinfo {author} {\bibfnamefont {H.}~\bibnamefont {Yuan}},\ }\bibfield  {title} {\bibinfo {title} {Minimal {{Trade-Off}} and {{Optimal Measurement}} for {{Multiparameter Quantum Estimation}}},\ }\href {https://doi.org/10.1103/dvj7-9thf} {\bibfield  {journal} {\bibinfo  {journal} {Phys. Rev. Lett.}\ }\textbf {\bibinfo {volume} {137}},\ \bibinfo {pages} {020804} (\bibinfo {year} {2026})},\ \Eprint {https://arxiv.org/abs/2605.23514} {arXiv:2605.23514} \BibitemShut {NoStop}%
\bibitem [{\citenamefont {Albarelli}\ \emph {et~al.}(2026)\citenamefont {Albarelli}, \citenamefont {Branford},\ and\ \citenamefont {Rubio}}]{Albarelli2026c}%
  \BibitemOpen
  \bibfield  {author} {\bibinfo {author} {\bibfnamefont {F.}~\bibnamefont {Albarelli}}, \bibinfo {author} {\bibfnamefont {D.}~\bibnamefont {Branford}},\ and\ \bibinfo {author} {\bibfnamefont {J.}~\bibnamefont {Rubio}},\ }\bibfield  {title} {\bibinfo {title} {Measurement incompatibility in {{Bayesian}} multiparameter quantum estimation},\ }\href {https://doi.org/10.22331/q-2026-08-13-2192} {\bibfield  {journal} {\bibinfo  {journal} {Quantum}\ }\textbf {\bibinfo {volume} {10}},\ \bibinfo {pages} {2192} (\bibinfo {year} {2026})},\ \Eprint {https://arxiv.org/abs/2511.16645} {arXiv:2511.16645} \BibitemShut {NoStop}%
\bibitem [{\citenamefont {Rubio}\ and\ \citenamefont {Dunningham}(2020)}]{Rubio2020a}%
  \BibitemOpen
  \bibfield  {author} {\bibinfo {author} {\bibfnamefont {J.}~\bibnamefont {Rubio}}\ and\ \bibinfo {author} {\bibfnamefont {J.~A.}\ \bibnamefont {Dunningham}},\ }\bibfield  {title} {\bibinfo {title} {Bayesian multiparameter quantum metrology with limited data},\ }\href {https://doi.org/10.1103/PhysRevA.101.032114} {\bibfield  {journal} {\bibinfo  {journal} {Phys. Rev. A}\ }\textbf {\bibinfo {volume} {101}},\ \bibinfo {pages} {032114} (\bibinfo {year} {2020})},\ \Eprint {https://arxiv.org/abs/1906.04123} {arXiv:1906.04123} \BibitemShut {NoStop}%
\bibitem [{\citenamefont {Sidhu}\ and\ \citenamefont {Kok}(2020)}]{Sidhu2020}%
  \BibitemOpen
  \bibfield  {author} {\bibinfo {author} {\bibfnamefont {J.~S.}\ \bibnamefont {Sidhu}}\ and\ \bibinfo {author} {\bibfnamefont {P.}~\bibnamefont {Kok}},\ }\bibfield  {title} {\bibinfo {title} {Geometric perspective on quantum parameter estimation},\ }\href {https://doi.org/10.1116/1.5119961} {\bibfield  {journal} {\bibinfo  {journal} {AVS Quantum Sci.}\ }\textbf {\bibinfo {volume} {2}},\ \bibinfo {pages} {014701} (\bibinfo {year} {2020})},\ \Eprint {https://arxiv.org/abs/1907.06628} {arXiv:1907.06628} \BibitemShut {NoStop}%
\bibitem [{\citenamefont {Helstrom}(1976)}]{Helstrom1976}%
  \BibitemOpen
  \bibfield  {author} {\bibinfo {author} {\bibfnamefont {C.~W.}\ \bibnamefont {Helstrom}},\ }\href@noop {} {\emph {\bibinfo {title} {Quantum Detection and Estimation Theory}}}\ (\bibinfo  {publisher} {Academic Press},\ \bibinfo {address} {New York},\ \bibinfo {year} {1976})\BibitemShut {NoStop}%
\bibitem [{\citenamefont {Holevo}(2011)}]{Holevo2011b}%
  \BibitemOpen
  \bibfield  {author} {\bibinfo {author} {\bibfnamefont {A.~S.}\ \bibnamefont {Holevo}},\ }\href {http://link.springer.com/10.1007/978-88-7642-378-9} {\emph {\bibinfo {title} {Probabilistic and {{Statistical Aspects}} of {{Quantum Theory}}}}},\ \bibinfo {edition} {2nd}\ ed.\ (\bibinfo  {publisher} {Edizioni della Normale},\ \bibinfo {address} {Pisa},\ \bibinfo {year} {2011})\BibitemShut {NoStop}%
\bibitem [{\citenamefont {Paris}(2009)}]{Paris2009}%
  \BibitemOpen
  \bibfield  {author} {\bibinfo {author} {\bibfnamefont {M.~G.~A.}\ \bibnamefont {Paris}},\ }\bibfield  {title} {\bibinfo {title} {Quantum estimation for quantum technology},\ }\href {https://doi.org/10.1142/S0219749909004839} {\bibfield  {journal} {\bibinfo  {journal} {Int. J. Quantum Inf.}\ }\textbf {\bibinfo {volume} {07}},\ \bibinfo {pages} {125} (\bibinfo {year} {2009})},\ \Eprint {https://arxiv.org/abs/0804.2981} {arXiv:0804.2981} \BibitemShut {NoStop}%
\bibitem [{\citenamefont {Liu}\ \emph {et~al.}(2020)\citenamefont {Liu}, \citenamefont {Yuan}, \citenamefont {Lu},\ and\ \citenamefont {Wang}}]{Liu2020}%
  \BibitemOpen
  \bibfield  {author} {\bibinfo {author} {\bibfnamefont {J.}~\bibnamefont {Liu}}, \bibinfo {author} {\bibfnamefont {H.}~\bibnamefont {Yuan}}, \bibinfo {author} {\bibfnamefont {X.-M.}\ \bibnamefont {Lu}},\ and\ \bibinfo {author} {\bibfnamefont {X.}~\bibnamefont {Wang}},\ }\bibfield  {title} {\bibinfo {title} {Quantum {{Fisher}} information matrix and multiparameter estimation},\ }\href {https://doi.org/10.1088/1751-8121/ab5d4d} {\bibfield  {journal} {\bibinfo  {journal} {J. Phys. A}\ }\textbf {\bibinfo {volume} {53}},\ \bibinfo {pages} {023001} (\bibinfo {year} {2020})},\ \Eprint {https://arxiv.org/abs/1907.08037} {arXiv:1907.08037} \BibitemShut {NoStop}%
\bibitem [{\citenamefont {Suzuki}(2024)}]{Suzuki2024a}%
  \BibitemOpen
  \bibfield  {author} {\bibinfo {author} {\bibfnamefont {J.}~\bibnamefont {Suzuki}},\ }\bibfield  {title} {\bibinfo {title} {Bayesian {{Nagaoka-Hayashi Bound}} for {{Multiparameter Quantum-State Estimation Problem}}},\ }\href {https://doi.org/10.1587/transfun.2023TAP0014} {\bibfield  {journal} {\bibinfo  {journal} {IEICE Trans. Fundam. Electron. Commun. Comput. Sci.}\ }\textbf {\bibinfo {volume} {E107.A}},\ \bibinfo {pages} {510} (\bibinfo {year} {2024})},\ \Eprint {https://arxiv.org/abs/2302.14223} {arXiv:2302.14223} \BibitemShut {NoStop}%
\bibitem [{\citenamefont {Zhang}\ and\ \citenamefont {Suzuki}(2024)}]{Zhang2024g}%
  \BibitemOpen
  \bibfield  {author} {\bibinfo {author} {\bibfnamefont {J.}~\bibnamefont {Zhang}}\ and\ \bibinfo {author} {\bibfnamefont {J.}~\bibnamefont {Suzuki}},\ }\href {https://doi.org/10.48550/arXiv.2405.10525} {\bibinfo {title} {Bayesian {{Logarithmic Derivative Type Lower Bounds}} for {{Quantum Estimation}}}} (\bibinfo {year} {2024}),\ \Eprint {https://arxiv.org/abs/2405.10525} {arXiv:2405.10525} \BibitemShut {NoStop}%
\bibitem [{\citenamefont {Zhang}\ \emph {et~al.}(2026)\citenamefont {Zhang}, \citenamefont {Yamagata},\ and\ \citenamefont {Suzuki}}]{Zhang2026n}%
  \BibitemOpen
  \bibfield  {author} {\bibinfo {author} {\bibfnamefont {J.}~\bibnamefont {Zhang}}, \bibinfo {author} {\bibfnamefont {K.}~\bibnamefont {Yamagata}},\ and\ \bibinfo {author} {\bibfnamefont {J.}~\bibnamefont {Suzuki}},\ }\href {https://doi.org/10.48550/arXiv.2607.01685} {\bibinfo {title} {Bayesian {{Monotone Metrics}} for {{Multiparameter Quantum Estimation}}}} (\bibinfo {year} {2026}),\ \Eprint {https://arxiv.org/abs/2607.01685} {arXiv:2607.01685} \BibitemShut {NoStop}%
\bibitem [{\citenamefont {Zhao}\ \emph {et~al.}(2026)\citenamefont {Zhao}, \citenamefont {Yamagata},\ and\ \citenamefont {Suzuki}}]{Zhao2026c}%
  \BibitemOpen
  \bibfield  {author} {\bibinfo {author} {\bibfnamefont {K.-H.}\ \bibnamefont {Zhao}}, \bibinfo {author} {\bibfnamefont {K.}~\bibnamefont {Yamagata}},\ and\ \bibinfo {author} {\bibfnamefont {J.}~\bibnamefont {Suzuki}},\ }\href {https://doi.org/10.48550/arXiv.2607.07031} {\bibinfo {title} {Bayesian {{Gill-Massar Bound}}: {{An Attainable Lower Bound}} for {{Qubit Parameter Estimation}}}} (\bibinfo {year} {2026}),\ \Eprint {https://arxiv.org/abs/2607.07031} {arXiv:2607.07031} \BibitemShut {NoStop}%
\bibitem [{\citenamefont {Bia}\ and\ \citenamefont {Gagatsos}(2026)}]{Bia2026}%
  \BibitemOpen
  \bibfield  {author} {\bibinfo {author} {\bibfnamefont {L.}~\bibnamefont {Bia}}\ and\ \bibinfo {author} {\bibfnamefont {C.~N.}\ \bibnamefont {Gagatsos}},\ }\href {https://doi.org/10.48550/arXiv.2607.15398} {\bibinfo {title} {Saturating the {{Bayesian Nagaoka-Hayashi}} bound within numerical precision for the depolarization {{SU}}(2) rotation channel}} (\bibinfo {year} {2026}),\ \Eprint {https://arxiv.org/abs/2607.15398} {arXiv:2607.15398} \BibitemShut {NoStop}%
\bibitem [{\citenamefont {Law}\ and\ \citenamefont {Bern{\'a}d}(2026)}]{Law2026}%
  \BibitemOpen
  \bibfield  {author} {\bibinfo {author} {\bibfnamefont {C.~K.~D.}\ \bibnamefont {Law}}\ and\ \bibinfo {author} {\bibfnamefont {J.~Z.}\ \bibnamefont {Bern{\'a}d}},\ }\href {https://doi.org/10.48550/arXiv.2608.10230} {\bibinfo {title} {Bayesian approach to simultaneous quantum multiparameter estimation with finite data}} (\bibinfo {year} {2026}),\ \Eprint {https://arxiv.org/abs/2608.10230} {arXiv:2608.10230} \BibitemShut {NoStop}%
\bibitem [{\citenamefont {Yang}\ \emph {et~al.}(2019)\citenamefont {Yang}, \citenamefont {Chiribella},\ and\ \citenamefont {Hayashi}}]{Yang2019}%
  \BibitemOpen
  \bibfield  {author} {\bibinfo {author} {\bibfnamefont {Y.}~\bibnamefont {Yang}}, \bibinfo {author} {\bibfnamefont {G.}~\bibnamefont {Chiribella}},\ and\ \bibinfo {author} {\bibfnamefont {M.}~\bibnamefont {Hayashi}},\ }\bibfield  {title} {\bibinfo {title} {Attaining the {{Ultimate Precision Limit}} in {{Quantum State Estimation}}},\ }\href {https://doi.org/10.1007/s00220-019-03433-4} {\bibfield  {journal} {\bibinfo  {journal} {Commun. Math. Phys.}\ }\textbf {\bibinfo {volume} {368}},\ \bibinfo {pages} {223} (\bibinfo {year} {2019})},\ \Eprint {https://arxiv.org/abs/1802.07587} {arXiv:1802.07587} \BibitemShut {NoStop}%
\bibitem [{\citenamefont {Carollo}\ \emph {et~al.}(2019)\citenamefont {Carollo}, \citenamefont {Spagnolo}, \citenamefont {Dubkov},\ and\ \citenamefont {Valenti}}]{Carollo2019}%
  \BibitemOpen
  \bibfield  {author} {\bibinfo {author} {\bibfnamefont {A.}~\bibnamefont {Carollo}}, \bibinfo {author} {\bibfnamefont {B.}~\bibnamefont {Spagnolo}}, \bibinfo {author} {\bibfnamefont {A.~A.}\ \bibnamefont {Dubkov}},\ and\ \bibinfo {author} {\bibfnamefont {D.}~\bibnamefont {Valenti}},\ }\bibfield  {title} {\bibinfo {title} {On quantumness in multi-parameter quantum estimation},\ }\href {https://doi.org/10.1088/1742-5468/ab3ccb} {\bibfield  {journal} {\bibinfo  {journal} {J. Stat. Mech. Theory Exp.}\ }\textbf {\bibinfo {volume} {2019}},\ \bibinfo {pages} {094010} (\bibinfo {year} {2019})},\ \Eprint {https://arxiv.org/abs/1911.11621} {arXiv:1911.11621} \BibitemShut {NoStop}%
\bibitem [{\citenamefont {Tsang}\ \emph {et~al.}(2020)\citenamefont {Tsang}, \citenamefont {Albarelli},\ and\ \citenamefont {Datta}}]{Tsang2020a}%
  \BibitemOpen
  \bibfield  {author} {\bibinfo {author} {\bibfnamefont {M.}~\bibnamefont {Tsang}}, \bibinfo {author} {\bibfnamefont {F.}~\bibnamefont {Albarelli}},\ and\ \bibinfo {author} {\bibfnamefont {A.}~\bibnamefont {Datta}},\ }\bibfield  {title} {\bibinfo {title} {Quantum {{Semiparametric Estimation}}},\ }\href {https://doi.org/10.1103/PhysRevX.10.031023} {\bibfield  {journal} {\bibinfo  {journal} {Phys. Rev. X}\ }\textbf {\bibinfo {volume} {10}},\ \bibinfo {pages} {031023} (\bibinfo {year} {2020})},\ \Eprint {https://arxiv.org/abs/1906.09871} {arXiv:1906.09871} \BibitemShut {NoStop}%
\bibitem [{\citenamefont {Rubio}(2024)}]{Rubio2024}%
  \BibitemOpen
  \bibfield  {author} {\bibinfo {author} {\bibfnamefont {J.}~\bibnamefont {Rubio}},\ }\bibfield  {title} {\bibinfo {title} {First-principles construction of symmetry-informed quantum metrologies},\ }\href {https://doi.org/10.1103/PhysRevA.110.L030401} {\bibfield  {journal} {\bibinfo  {journal} {Phys. Rev. A}\ }\textbf {\bibinfo {volume} {110}},\ \bibinfo {pages} {L030401} (\bibinfo {year} {2024})},\ \Eprint {https://arxiv.org/abs/2402.16410} {arXiv:2402.16410} \BibitemShut {NoStop}%
\bibitem [{\citenamefont {{Demkowicz-Dobrza{\'n}ski}}(2011)}]{Demkowicz-Dobrzanski2011}%
  \BibitemOpen
  \bibfield  {author} {\bibinfo {author} {\bibfnamefont {R.}~\bibnamefont {{Demkowicz-Dobrza{\'n}ski}}},\ }\bibfield  {title} {\bibinfo {title} {Optimal phase estimation with arbitrary a priori knowledge},\ }\href {https://doi.org/10.1103/PhysRevA.83.061802} {\bibfield  {journal} {\bibinfo  {journal} {Phys. Rev. A}\ }\textbf {\bibinfo {volume} {83}},\ \bibinfo {pages} {061802} (\bibinfo {year} {2011})},\ \Eprint {https://arxiv.org/abs/1102.0786} {arXiv:1102.0786} \BibitemShut {NoStop}%
\bibitem [{\citenamefont {Mishra}\ \emph {et~al.}(2025)\citenamefont {Mishra}, \citenamefont {Lami},\ and\ \citenamefont {Wilde}}]{Mishra2025}%
  \BibitemOpen
  \bibfield  {author} {\bibinfo {author} {\bibfnamefont {H.~K.}\ \bibnamefont {Mishra}}, \bibinfo {author} {\bibfnamefont {L.}~\bibnamefont {Lami}},\ and\ \bibinfo {author} {\bibfnamefont {M.~M.}\ \bibnamefont {Wilde}},\ }\href {https://doi.org/10.48550/arXiv.2505.20447} {\bibinfo {title} {Near-optimal performance of square-root measurement for general score functions and quantum ensembles}} (\bibinfo {year} {2025}),\ \Eprint {https://arxiv.org/abs/2505.20447} {arXiv:2505.20447} \BibitemShut {NoStop}%
\bibitem [{\citenamefont {Lin}\ and\ \citenamefont {de~Wolf}(2025)}]{Lin2025a}%
  \BibitemOpen
  \bibfield  {author} {\bibinfo {author} {\bibfnamefont {H.-H.}\ \bibnamefont {Lin}}\ and\ \bibinfo {author} {\bibfnamefont {R.}~\bibnamefont {de~Wolf}},\ }\href {https://doi.org/10.48550/arXiv.2506.01903} {\bibinfo {title} {Getting almost all the bits from a quantum random access code}} (\bibinfo {year} {2025}),\ \Eprint {https://arxiv.org/abs/2506.01903} {arXiv:2506.01903} \BibitemShut {NoStop}%
\bibitem [{\citenamefont {Grewal}\ \emph {et~al.}(2026)\citenamefont {Grewal}, \citenamefont {Pelecanos}, \citenamefont {Spilecki}, \citenamefont {Tang},\ and\ \citenamefont {Wright}}]{Grewal2026a}%
  \BibitemOpen
  \bibfield  {author} {\bibinfo {author} {\bibfnamefont {S.}~\bibnamefont {Grewal}}, \bibinfo {author} {\bibfnamefont {A.}~\bibnamefont {Pelecanos}}, \bibinfo {author} {\bibfnamefont {J.}~\bibnamefont {Spilecki}}, \bibinfo {author} {\bibfnamefont {E.}~\bibnamefont {Tang}},\ and\ \bibinfo {author} {\bibfnamefont {J.}~\bibnamefont {Wright}},\ }\href {https://doi.org/10.48550/arXiv.2610.06781} {\bibinfo {title} {Pretty good shadow tomography}} (\bibinfo {year} {2026}),\ \Eprint {https://arxiv.org/abs/2610.06781} {arXiv:2610.06781} \BibitemShut {NoStop}%
\bibitem [{Note1()}]{Note1}%
  \BibitemOpen
  \bibinfo {note} {The generalised BK inequality on which our results are based holds also for infinite-dimensional Hilbert spaces~\cite {Mishra2025}. Thus, we expect our results to extend to this setting under mild regularity assumptions, which however we do not work out in detail here, favouring a more concise and pedagogical presentation.}\BibitemShut {Stop}%
\bibitem [{\citenamefont {Suzuki}\ \emph {et~al.}(2020)\citenamefont {Suzuki}, \citenamefont {Yang},\ and\ \citenamefont {Hayashi}}]{Suzuki2020b}%
  \BibitemOpen
  \bibfield  {author} {\bibinfo {author} {\bibfnamefont {J.}~\bibnamefont {Suzuki}}, \bibinfo {author} {\bibfnamefont {Y.}~\bibnamefont {Yang}},\ and\ \bibinfo {author} {\bibfnamefont {M.}~\bibnamefont {Hayashi}},\ }\bibfield  {title} {\bibinfo {title} {Quantum state estimation with nuisance parameters},\ }\href {https://doi.org/10.1088/1751-8121/ab8b78} {\bibfield  {journal} {\bibinfo  {journal} {J. Phys. A}\ }\textbf {\bibinfo {volume} {53}},\ \bibinfo {pages} {453001} (\bibinfo {year} {2020})},\ \Eprint {https://arxiv.org/abs/1911.02790} {arXiv:1911.02790} \BibitemShut {NoStop}%
\bibitem [{\citenamefont {Personick}(1971)}]{Personick1971}%
  \BibitemOpen
  \bibfield  {author} {\bibinfo {author} {\bibfnamefont {S.~D.}\ \bibnamefont {Personick}},\ }\bibfield  {title} {\bibinfo {title} {Application of quantum estimation theory to analog communication over quantum channels},\ }\href {https://doi.org/10.1109/TIT.1971.1054643} {\bibfield  {journal} {\bibinfo  {journal} {IEEE Trans. Inf. Theory}\ }\textbf {\bibinfo {volume} {17}},\ \bibinfo {pages} {240} (\bibinfo {year} {1971})}\BibitemShut {NoStop}%
\bibitem [{\citenamefont {Belavkin}(1975)}]{Belavkin1975}%
  \BibitemOpen
  \bibfield  {author} {\bibinfo {author} {\bibfnamefont {V.~P.}\ \bibnamefont {Belavkin}},\ }\bibfield  {title} {\bibinfo {title} {Optimal multiple quantum statistical hypothesis testing},\ }\href {https://doi.org/10.1080/17442507508833114} {\bibfield  {journal} {\bibinfo  {journal} {Stochastics}\ }\textbf {\bibinfo {volume} {1}},\ \bibinfo {pages} {315} (\bibinfo {year} {1975})}\BibitemShut {NoStop}%
\bibitem [{\citenamefont {Hausladen}\ and\ \citenamefont {Wootters}(1994)}]{Hausladen1994}%
  \BibitemOpen
  \bibfield  {author} {\bibinfo {author} {\bibfnamefont {P.}~\bibnamefont {Hausladen}}\ and\ \bibinfo {author} {\bibfnamefont {W.~K.}\ \bibnamefont {Wootters}},\ }\bibfield  {title} {\bibinfo {title} {A `{{Pretty Good}}' {{Measurement}} for {{Distinguishing Quantum States}}},\ }\href {https://doi.org/10.1080/09500349414552221} {\bibfield  {journal} {\bibinfo  {journal} {J. Mod. Opt.}\ }\textbf {\bibinfo {volume} {41}},\ \bibinfo {pages} {2385} (\bibinfo {year} {1994})}\BibitemShut {NoStop}%
\bibitem [{\citenamefont {Macchiavello}(2003)}]{Macchiavello2003}%
  \BibitemOpen
  \bibfield  {author} {\bibinfo {author} {\bibfnamefont {C.}~\bibnamefont {Macchiavello}},\ }\bibfield  {title} {\bibinfo {title} {Optimal estimation of multiple phases},\ }\href {https://doi.org/10.1103/PhysRevA.67.062302} {\bibfield  {journal} {\bibinfo  {journal} {Phys. Rev. A}\ }\textbf {\bibinfo {volume} {67}},\ \bibinfo {pages} {062302} (\bibinfo {year} {2003})},\ \Eprint {https://arxiv.org/abs/quant-ph/0304126} {arXiv:quant-ph/0304126} \BibitemShut {NoStop}%
\bibitem [{\citenamefont {Barbieri}\ \emph {et~al.}(2024)\citenamefont {Barbieri}, \citenamefont {Gianani}, \citenamefont {Goldberg},\ and\ \citenamefont {{S{\'a}nchez-Soto}}}]{Barbieri2024a}%
  \BibitemOpen
  \bibfield  {author} {\bibinfo {author} {\bibfnamefont {M.}~\bibnamefont {Barbieri}}, \bibinfo {author} {\bibfnamefont {I.}~\bibnamefont {Gianani}}, \bibinfo {author} {\bibfnamefont {A.~Z.}\ \bibnamefont {Goldberg}},\ and\ \bibinfo {author} {\bibfnamefont {L.~L.}\ \bibnamefont {{S{\'a}nchez-Soto}}},\ }\bibfield  {title} {\bibinfo {title} {Quantum multiphase estimation},\ }\href {https://doi.org/10.1080/00107514.2025.2469974} {\bibfield  {journal} {\bibinfo  {journal} {Contemp. Phys.}\ }\textbf {\bibinfo {volume} {65}},\ \bibinfo {pages} {112} (\bibinfo {year} {2024})},\ \Eprint {https://arxiv.org/abs/2502.07873} {arXiv:2502.07873} \BibitemShut {NoStop}%
\bibitem [{\citenamefont {{Demkowicz-Dobrza{\'n}ski}}\ \emph {et~al.}(2015)\citenamefont {{Demkowicz-Dobrza{\'n}ski}}, \citenamefont {Jarzyna},\ and\ \citenamefont {Ko{\l}ody{\'n}ski}}]{Demkowicz-Dobrzanski2015a}%
  \BibitemOpen
  \bibfield  {author} {\bibinfo {author} {\bibfnamefont {R.}~\bibnamefont {{Demkowicz-Dobrza{\'n}ski}}}, \bibinfo {author} {\bibfnamefont {M.}~\bibnamefont {Jarzyna}},\ and\ \bibinfo {author} {\bibfnamefont {J.}~\bibnamefont {Ko{\l}ody{\'n}ski}},\ }\bibfield  {title} {\bibinfo {title} {Quantum {{Limits}} in {{Optical Interferometry}}},\ }in\ \href {https://doi.org/10.1016/bs.po.2015.02.003} {\emph {\bibinfo {booktitle} {Progress in {{Optics}}, {{Volume}} 60}}},\ \bibinfo {editor} {edited by\ \bibinfo {editor} {\bibfnamefont {E.}~\bibnamefont {Wolf}}}\ (\bibinfo  {publisher} {Elsevier},\ \bibinfo {address} {Amsterdam},\ \bibinfo {year} {2015})\ pp.\ \bibinfo {pages} {345--435},\ \Eprint {https://arxiv.org/abs/1405.7703} {arXiv:1405.7703} \BibitemShut {NoStop}%
\bibitem [{\citenamefont {Collett}(1993)}]{Collett1993}%
  \BibitemOpen
  \bibfield  {author} {\bibinfo {author} {\bibfnamefont {M.~J.}\ \bibnamefont {Collett}},\ }\bibfield  {title} {\bibinfo {title} {Phase noise in a squeezed state},\ }\href {https://doi.org/10.1088/0031-8949/1993/T48/019} {\bibfield  {journal} {\bibinfo  {journal} {Phys. Scr.}\ }\textbf {\bibinfo {volume} {T48}},\ \bibinfo {pages} {124} (\bibinfo {year} {1993})}\BibitemShut {NoStop}%
\bibitem [{\citenamefont {Berry}\ \emph {et~al.}(2009)\citenamefont {Berry}, \citenamefont {Higgins}, \citenamefont {Bartlett}, \citenamefont {Mitchell}, \citenamefont {Pryde},\ and\ \citenamefont {Wiseman}}]{Berry2009}%
  \BibitemOpen
  \bibfield  {author} {\bibinfo {author} {\bibfnamefont {D.~W.}\ \bibnamefont {Berry}}, \bibinfo {author} {\bibfnamefont {B.~L.}\ \bibnamefont {Higgins}}, \bibinfo {author} {\bibfnamefont {S.~D.}\ \bibnamefont {Bartlett}}, \bibinfo {author} {\bibfnamefont {M.~W.}\ \bibnamefont {Mitchell}}, \bibinfo {author} {\bibfnamefont {G.~J.}\ \bibnamefont {Pryde}},\ and\ \bibinfo {author} {\bibfnamefont {H.~M.}\ \bibnamefont {Wiseman}},\ }\bibfield  {title} {\bibinfo {title} {How to perform the most accurate possible phase measurements},\ }\href {https://doi.org/10.1103/PhysRevA.80.052114} {\bibfield  {journal} {\bibinfo  {journal} {Phys. Rev. A}\ }\textbf {\bibinfo {volume} {80}},\ \bibinfo {pages} {052114} (\bibinfo {year} {2009})},\ \Eprint {https://arxiv.org/abs/0907.0014} {arXiv:0907.0014} \BibitemShut {NoStop}%
\bibitem [{\citenamefont {Berg}\ \emph {et~al.}(1984)\citenamefont {Berg}, \citenamefont {Christensen},\ and\ \citenamefont {Ressel}}]{Berg1984}%
  \BibitemOpen
  \bibfield  {author} {\bibinfo {author} {\bibfnamefont {C.}~\bibnamefont {Berg}}, \bibinfo {author} {\bibfnamefont {J.~P.~R.}\ \bibnamefont {Christensen}},\ and\ \bibinfo {author} {\bibfnamefont {P.}~\bibnamefont {Ressel}},\ }\href@noop {} {\emph {\bibinfo {title} {Harmonic Analysis on Semigroups: Theory of Positive Definite and Related Functions}}},\ \bibinfo {series} {Graduate Texts in Mathematics}\ No.\ \bibinfo {number} {100}\ (\bibinfo  {publisher} {Springer-Verlag},\ \bibinfo {address} {New York},\ \bibinfo {year} {1984})\BibitemShut {NoStop}%
\bibitem [{\citenamefont {Jaynes}\ and\ \citenamefont {Bretthorst}(2003)}]{Jaynes2003a}%
  \BibitemOpen
  \bibfield  {author} {\bibinfo {author} {\bibfnamefont {E.~T.}\ \bibnamefont {Jaynes}}\ and\ \bibinfo {author} {\bibfnamefont {G.~L.}\ \bibnamefont {Bretthorst}},\ }\href@noop {} {\emph {\bibinfo {title} {Probability Theory: The Logic of Science}}}\ (\bibinfo  {publisher} {Cambridge University Press},\ \bibinfo {address} {Cambridge},\ \bibinfo {year} {2003})\BibitemShut {NoStop}%
\bibitem [{\citenamefont {Ragy}\ \emph {et~al.}(2016)\citenamefont {Ragy}, \citenamefont {Jarzyna},\ and\ \citenamefont {{Demkowicz-Dobrza{\'n}ski}}}]{Ragy2016}%
  \BibitemOpen
  \bibfield  {author} {\bibinfo {author} {\bibfnamefont {S.}~\bibnamefont {Ragy}}, \bibinfo {author} {\bibfnamefont {M.}~\bibnamefont {Jarzyna}},\ and\ \bibinfo {author} {\bibfnamefont {R.}~\bibnamefont {{Demkowicz-Dobrza{\'n}ski}}},\ }\bibfield  {title} {\bibinfo {title} {Compatibility in multiparameter quantum metrology},\ }\href {https://doi.org/10.1103/PhysRevA.94.052108} {\bibfield  {journal} {\bibinfo  {journal} {Phys. Rev. A}\ }\textbf {\bibinfo {volume} {94}},\ \bibinfo {pages} {052108} (\bibinfo {year} {2016})},\ \Eprint {https://arxiv.org/abs/1608.02634} {arXiv:1608.02634} \BibitemShut {NoStop}%
\bibitem [{Note2()}]{Note2}%
  \BibitemOpen
  \bibinfo {note} {R.~Demkowicz-Dobrza\'nski, private communication; the argument will appear in a forthcoming book.}\BibitemShut {Stop}%
\bibitem [{\citenamefont {Meyer}\ \emph {et~al.}(2025)\citenamefont {Meyer}, \citenamefont {Khatri}, \citenamefont {Stilck~Fran{\c c}a}, \citenamefont {Eisert},\ and\ \citenamefont {Faist}}]{Meyer2025}%
  \BibitemOpen
  \bibfield  {author} {\bibinfo {author} {\bibfnamefont {J.~J.}\ \bibnamefont {Meyer}}, \bibinfo {author} {\bibfnamefont {S.}~\bibnamefont {Khatri}}, \bibinfo {author} {\bibfnamefont {D.}~\bibnamefont {Stilck~Fran{\c c}a}}, \bibinfo {author} {\bibfnamefont {J.}~\bibnamefont {Eisert}},\ and\ \bibinfo {author} {\bibfnamefont {P.}~\bibnamefont {Faist}},\ }\bibfield  {title} {\bibinfo {title} {Quantum {{Metrology}} in the {{Finite-Sample Regime}}},\ }\href {https://doi.org/10.1103/qbn1-p6bq} {\bibfield  {journal} {\bibinfo  {journal} {PRX Quantum}\ }\textbf {\bibinfo {volume} {6}},\ \bibinfo {pages} {030336} (\bibinfo {year} {2025})},\ \Eprint {https://arxiv.org/abs/2307.06370} {arXiv:2307.06370} \BibitemShut {NoStop}%
\bibitem [{\citenamefont {Belliardo}\ \emph {et~al.}(2026)\citenamefont {Belliardo}, \citenamefont {Gardner}, \citenamefont {Jiang},\ and\ \citenamefont {Clerk}}]{Belliardo2026b}%
  \BibitemOpen
  \bibfield  {author} {\bibinfo {author} {\bibfnamefont {F.}~\bibnamefont {Belliardo}}, \bibinfo {author} {\bibfnamefont {J.~W.}\ \bibnamefont {Gardner}}, \bibinfo {author} {\bibfnamefont {L.}~\bibnamefont {Jiang}},\ and\ \bibinfo {author} {\bibfnamefont {A.~A.}\ \bibnamefont {Clerk}},\ }\href {https://doi.org/10.48550/arXiv.2608.15528} {\bibinfo {title} {Global precision bounds and success-probability guarantees in quantum parameter learning}} (\bibinfo {year} {2026}),\ \Eprint {https://arxiv.org/abs/2608.15528} {arXiv:2608.15528} \BibitemShut {NoStop}%
\end{thebibliography}%

\end{document}